\documentclass[letterpaper, 10 pt, conference]{ieeeconf}  

\IEEEoverridecommandlockouts                              
\newcommand{\norm}[1]{\left\lVert#1\right\rVert}

\usepackage{graphics} 
\usepackage{epsfig} 
\usepackage{mathptmx} 
\usepackage{times} 
\usepackage{amsmath} 
\usepackage{amssymb}  
\usepackage{mathtools}
\usepackage{algorithm}
\usepackage{dsfont}
\usepackage{algorithmicx}
\usepackage{algpascal}
\usepackage{textcomp}

\newtheorem{prop}{Proposition}

\title{\LARGE \bf
Infinite-Dimensional Sparse Learning in Linear System Identification
}
\author{Mingzhou Yin, Mehmet Tolga Akan, Andrea Iannelli, and Roy S. Smith
\thanks{This work was supported by the Swiss National Science Foundation under Grant 200021\_178890.}
\thanks{M. Yin, A. Iannelli, and R. S. Smith are with the Automatic Control Laboratory, Swiss Federal Institute of Technology (ETH Zurich), Physikstrasse 3, 8092 Zurich, Switzerland, {\tt\small \{myin,iannelli,rsmith\}} {\tt\small @control.ee.ethz.ch}.}%
\thanks{M. T. Akan is with Control System Group, Eindhoven University of Technology, 5600 MB Eindhoven, The Netherlands, {\tt\small m.t.akan@tue.nl.}}%
}

\begin{document}

\maketitle
\thispagestyle{empty}
\pagestyle{empty}

\noindent
\fbox{\begin{minipage}{\linewidth}

M. Yin, M. T. Akan, A. Iannelli and R. S. Smith, ``Infinite-Dimensional Sparse Learning in Linear System Identification," arXiv:2203.14731.

\vspace{0.5em}

\textcopyright\ 2022 IEEE. Personal use of this material is permitted. Permission from IEEE must be obtained for all other uses, in any current or future media, including reprinting/republishing this material for advertising or promotional purposes, creating new collective works, for resale or redistribution to servers or lists, or reuse of any copyrighted component of this work in other works.
\end{minipage}}
\vspace{1em}

\begin{abstract}
Regularized methods have been widely applied to system identification problems without known model structures. This paper presents an infinite-dimensional sparse learning algorithm based on atomic norm regularization. Atomic norm regularization decomposes the transfer function into first-order atomic models and solves a group lasso problem that selects a sparse set of poles and identifies the corresponding coefficients. The difficulty in solving the problem lies in the fact that there are an infinite number of possible atomic models. This work proposes a greedy algorithm that generates new candidate atomic models maximizing the violation of the optimality conditions of the existing problem. This algorithm is able to solve the infinite-dimensional group lasso problem with high precision. The algorithm is further extended to reduce the bias and reject false positives in pole location estimation by iteratively reweighted adaptive group lasso and complementary pairs stability selection respectively. Numerical results demonstrate that the proposed algorithm performs better than benchmark parameterized and regularized methods in terms of both impulse response fitting and pole location estimation.
\end{abstract}

\section{Introduction}
System identification investigates the problem of identifying models of dynamical systems from measured input-output data. This problem has been widely studied under the parameter estimation framework, where the system is modeled by a finite-dimensional parametrization \cite{LjungBook2}. The optimal model parameters can then be estimated by tools in classical statistics. One well-known approach in this category is the prediction error method (PEM) based on maximum likelihood estimation \cite{ASTROM1980551}.

However, such approaches only work when model structure and complexity are known, and the associated optimization problems are only convex for particular noise models, e.g., ARX models \cite{Ljung_2019}. Alternative approaches have been proposed in the last decade, which identify general high-dimensional models with regularization techniques to encode prior model knowledge \cite{Ljung2_2019,Pillonetto_2016}. In particular, kernel-based identification \cite{PILLONETTO201081,PILLONETTO2014657} has received significant attention, whereby, in its basic form, a truncated impulse response model is identified with a Tikhonov regularization term. The performance of this approach depends heavily on the choice of kernels which need to be carefully designed \cite{CHEN2018109}. This kernel design step poses similar problems as model structure selection in the classical paradigm. In addition, kernel-based identification controls model complexity through the norm of the impulse response induced by an arbitrary reproducing kernel Hilbert space \cite{Chen_2012}. Such complexity measures do not have clear system theoretic interpretations.

Alternative regularization approaches have also been proposed to directly control the number of poles of the model. This measure has a more concrete meaning for system analysis and control, either when the system is known to have a low-order structure, or when a low-order representation is desired. The Hankel nuclear norm of the impulse response is used as a convex surrogate in \cite{Smith_2014,945730}. However, this regularizer is known to be prone to stability issues \cite{Pillonetto_2016}. A different approach consists of modeling the system as a summation of first-order ``atoms", which are some predefined basis models. The model complexity can then be controlled by regularizing the $l_1$-norm of the coefficients. This is known as regularizing the atomic norm with respect to the atomic decomposition \cite{6426006}. This results in a lasso-type problem that promotes models with a small number of poles \cite{Yuan_2006}. This idea has also been used in periodic system identification \cite{YIN20201237} and kernel design \cite{KHOSRAVI2020412}. Another advantage of the first-order atomic decomposition is that it directly identifies the pole locations of the system. Pole locations are important in classical control design, yet hard to estimate with conventional identification approaches.

Existing work on the atomic norm regularization approach, however, has multiple known drawbacks. First, instead of solving the group lasso problem on an infinite set of stable atoms, only a finite discretization of the atomic set is considered for tractability. This leads to an approximation error which can only be reduced with a very large set of atoms \cite{6426006}. In addition, a large bias is induced by lasso-type regularization \cite{Pillonetto_2016}, and the pole location estimation contains a possibly large number of false positives due to the ``p-value lottery" in high-dimensional regression \cite{doi:10.1198/jasa.2009.tm08647}.

In this paper, we propose an infinite-dimensional sparse learning algorithm based on atomic norm regularization, which aims to tackle the above drawbacks. This algorithm directly targets the group lasso problem with an infinite feature set, which has been studied in the machine learning literature \cite{Rakotomamonjy_2012,10.1007/978-3-540-72927-3_39,NIPS2014_459a4ddc}. Similar to Algorithm 1 in \cite{Rakotomamonjy_2012}, our proposed algorithm first solves the problem with a small number of randomly generated features. Then, 
a new atomic model feature is selected to maximize the optimality condition violation for the previous iteration. 
The algorithm guarantees a decrease in the objective value per iteration and solves the infinite-dimensional problem with an arbitrarily small tolerance.

Two different strategies are further presented to debias the estimate and reject false positives respectively. Iteratively reweighted adaptive group lasso \cite{WANG20085277,5109694} is applied to reduce the amount of regularization on significant modes of the identified model, and thus reduce the bias. Complementary pairs stability selection (CPSS) \cite{HighDimensional,https://doi.org/10.1111/j.1467-9868.2011.01034.x} solves the problem repeatedly on subsamples of the identification data and estimates the pole location by selecting atoms that are consistently active.

Numerical results demonstrate that the proposed algorithm performs better than PEM with an ARX model, kernel-based identification with tuned/correlated (TC) kernel design, and the existing atomic norm regularization algorithm in terms of impulse response fitting on a benchmark system. In addition, adaptive group lasso is able to reduce the bias of the algorithm and CPSS obtains more accurate pole location estimation than PEM with fewer false positives.

\section{Atomic Norm Regularization in System Identification}
In this work, we consider a strictly causal and stable linear time-invariant single-input single-output discrete-time system
$
    y(t)=G_0(q)u(t)+v(t)
$, 
where $u(t)$, $y(t)$, $v(t)$ are the inputs, outputs and additive noise respectively, and $q$ is the shift operator. The transfer function $G_0(q)$ is assumed to have a low number of poles. The additive noise is assumed to be zero-mean i.i.d. Gaussian with a variance of $\sigma^2$. An input-output sequence of the system
\begin{equation}
    \mathbf{u}=\left[u(1)\ u(2)\ \dots\ u(N)\right]^\top\!,\  \mathbf{y}=\left[y(1)\ y(2)\ \dots\ y(N)\right]^\top
\end{equation}
has been collected. We are interested in identifying the transfer function $G_0(q)$ from the data sequence $(\mathbf{u},\mathbf{y})$.

In regularized system identification, the transfer function $G_0(q)$ is expressed with a general high-dimensional parametrization
$
    G_0(q)=\sum_{k\in K} c_k A_k(q)
$, 
where $A_k(q)$ are the basis transfer functions known as atoms \cite{6426006}, $c_k$ are the corresponding coefficients, and $K$ denotes the set of indices. Denote the set of coefficients as $C=\left\{c_k|\,k\in K\right\}$. The following regularized optimization problem is solved:
\begin{equation}
    \underset{C}{\text{minimize}}\quad V\left(\mathbf{y}-\sum_{k\in K}c_k\, \phi\!\left(A_k(q),\mathbf{u}\right)\right) + \lambda J\left(C\right),
    \label{eqn:1}
\end{equation}
where $\phi\!\left(A(q),\mathbf{u}\right)$ denotes the length-$N$ output response of the system $A(q)$ to the inputs $\mathbf{u}$, $V(\cdot)$ is the loss function that penalizes the output residuals, $J(\cdot)$ is the regularization term that encodes prior knowledge of the coefficients, and $\lambda$ is the regularization parameter to tune the amount of regularization. For the rest of the paper, the loss function is selected as $V(x)=\norm{x}_2^2$, which is related to the maximum likelihood estimator when the noise $v(t)$ is i.i.d. Gaussian.

In this paper, the atomic decomposition of the transfer function in \cite{6426006} is employed, where
$
    A_k(q)=\dfrac{1-|k|^2}{q-k}
$, 
and the corresponding coefficients $c_k$ are complex numbers. Unlike conventional parametrizations, here $k$ is a stable pole within the open unit disk. The set of indices is thus
\begin{equation}
    K=\left\{k=\alpha\exp{\left(j\beta\right)}\,|\,\alpha\in[0,1),\beta\in[0,2\pi)\right\},
\end{equation}
which has infinite elements. The atoms $A_k(q)$ are normalized to have a Hankel nuclear norm of 1. Define the pole locations of the system as $S=\left\{k\,|\left|c_k\right|>0\right\}$, which is also known as the active atomic set. Since the system is known to have a small number of poles, a sparsity-promoting regularization term $J(C)$ is desired. In particular, an $l_1$-norm regularizer 
\begin{equation}
    J(C)=\sum_{k\in K}\left|c_k\right|
    \label{eqn:lasso}
\end{equation}
is used and defined as the atomic norm of the model \cite{10.2307/2346178}. Observe that for real-rational systems, the pole locations should be in conjugate pairs and the corresponding atomic responses are also complex conjugates of one another, i.e., $\phi\!\left(A_{\bar{k}}(q),\mathbf{u}\right)=\bar{\phi}\!\left(A_k(q),\mathbf{u}\right)$, where the overbar denotes the complex conjugate. This means that coefficients for a conjugate pole pair should also be complex conjugates, i.e., $c_{\bar{k}}=\bar{c}_k$. Adding this constraint on the coefficients of (\ref{eqn:1}), the problem can be reformulated as
\begin{equation}
    \underset{\left\{c_k\right\}_{k\in\hat{K}}}{\text{minimize}}\quad \norm{\,\mathbf{y}-\sum_{k\in\hat{K}}\left(c_k\, \phi_k+\bar{c}_k\, \bar{\phi}_k\right)}_2^2 + 2\lambda \sum_{k\in\hat{K}}\left|c_k\right|,
    \label{eqn:2}
\end{equation}
where $\phi_k:=\phi\!\left(A_k(q),\mathbf{u}\right)$ and
\begin{equation}
    \hat{K}=\left\{k=\alpha\exp{\left(j\beta\right)}\,|\,\alpha\in[0,1),\beta\in[0,\pi]\right\}
\end{equation}
denotes the upper half of the open unit disk.

Using $\Re$ and $\Im$ to denote real and imaginary parts, let 
\begin{equation}
\gamma_k=\begin{bmatrix}
\Re\!\left(c_k\right)&\Im\!\left(c_k\right)
\end{bmatrix}
^\top
,\  \zeta_k=\begin{bmatrix}2\Re\!\left(\phi_k\right)\ -2\Im\!\left(\phi_k\right)\end{bmatrix}.
\label{eqn:10}
\end{equation}
Substituting \eqref{eqn:10} into \eqref{eqn:2}, (\ref{eqn:2}) can be expressed as a real-valued problem,
\begin{equation}
    \Gamma^\star:=\left\{\gamma_k^\star\right\}_{k\in\hat{K}}=\underset{\left\{\gamma_k\right\}_{k\in\hat{K}}}{\text{argmin}}\ \underbrace{\norm{\,\mathbf{y}-\sum_{k\in\hat{K}}\zeta_k \gamma_k}_2^2 + 2\lambda \sum_{k\in\hat{K}}\norm{\gamma_k}_2}_{J(\Gamma)},
    \label{eqn:3}
\end{equation}
where $\Gamma:=\left\{\gamma_k\,|\,k\in\hat{K}\right\}$. Note that (\ref{eqn:3}) is a standard group lasso problem \cite{Yuan_2006}. The identified transfer function can be recovered by
\begin{equation}
    \hat{G}(q)=\sum_{k\in\hat{K}} \left[\,1\ \  j\,\right]\gamma_k^\star A_k(q) +  \left[\,1\ \  -j\,\right]\gamma_k^\star A_{\bar{k}}(q),
    \label{eqn:ghat}
\end{equation}
and the estimated pole locations are
\begin{equation}
    \hat{S}=\left\{k\,\left|\,\norm{\gamma^\star_k}_2>0\right.\right\}\cup\left\{\bar{k}\,\left|\,\norm{\gamma^\star_k}_2>0\right.\right\}.
    \label{eqn:phat}
\end{equation}
However, problem (\ref{eqn:3}) cannot be directly solved since it is an infinite-dimensional problem. Existing algorithms relax this problem by approximating $\hat{K}$ with a discrete grid \cite{6426006}. As shown in Proposition~4.1 of \cite{6426006}, the discretization induces a relative error in the atomic norm that is inversely proportional to the square root of the number of elements in the discretized $\hat{K}$.

\section{Algorithm for Infinite-Dimensional Atomic Norm Regularization Problems}
In this section, an algorithm is proposed to directly solve the infinite-dimensional problem (\ref{eqn:3}). This algorithm is inspired by the feature generation algorithm in \cite{Rakotomamonjy_2012}.

Problem (\ref{eqn:3}) is a non-differentiable convex program, whose optimality conditions are given by $0\in\partial J(\Gamma)$, where $\partial$ denotes the subdifferential. In detail, the optimality conditions of (\ref{eqn:3}) are
\begin{equation}
\begin{cases}
    \norm{\zeta_k^\top R}_2\leq\lambda,&\text{if }\norm{\gamma^\star_k}_2=0,\\
    \zeta_k^\top R+\lambda\gamma^\star_k/\norm{\gamma^\star_k}_2=0,&\text{if }\norm{\gamma^\star_k}_2> 0,
\end{cases}
\end{equation}
for all $k\in\hat{K}$, where $R:=\mathbf{y}-\sum_{k\in\hat{K}}\zeta_k \gamma^\star_k$ is the vector of output residuals. The derivation makes use of the property
\begin{equation}
    \partial \norm{\gamma^\star_k}_2=
\begin{cases}
    \left\{w\,|\,\norm{w}_2\leq 1\right\},&\norm{\gamma^\star_k}_2=0,\\
    \gamma^\star_k / \norm{\gamma^\star_k}_2, &\norm{\gamma^\star_k}_2> 0.
\end{cases}
\end{equation}

Let $\hat{K}_d=\left\{k_1,k_2,\dots,k_p\right\}$ be a finite subset of $\hat{K}$ with $p$ elements. Then, with an abuse of notation, by replacing $\hat{K}$ with $\hat{K}_d$ in (\ref{eqn:3}), a discretized optimal solution, denoted by $\Gamma^\star(\hat{K}_d):=\left\{\gamma^\star_i(\hat{K}_d)\right\}_{i=1}^p$, can be obtained, which satisfies
\begin{equation}
\begin{cases}
    \norm{\zeta_i(\hat{K}_d)^\top R(\hat{K}_d)}_2\leq\lambda,&\text{if }\norm{\gamma^\star_i(\hat{K}_d)}_2=0,\\
    \zeta_i(\hat{K}_d)^\top R(\hat{K}_d) +\lambda\dfrac{\gamma^\star_i(\hat{K}_d)}{\norm{\gamma^\star_i(\hat{K}_d)}_2}=0,&\text{if }\norm{\gamma^\star_i(\hat{K}_d)}_2> 0,
\end{cases}
\label{eqn:dopt}
\end{equation}
for $i=1,\dots,p$, where $R(\hat{K}_d):=\mathbf{y}-\sum_{i=1}^p\zeta_i(\hat{K}_d)\gamma^\star_i(\hat{K}_d)$ and $\zeta_i(\hat{K}_d):=\zeta_{k_i}$.

Suppose we want to add a new element $k_{p+1}$ to $\hat{K}_d$. Then the optimal solution with respect to $\hat{K}^+_d:=\hat{K}_d\cup \{k_{p+1}\}$ is
\begin{equation}
    \gamma^\star_i(\hat{K}^+_d)=
    \begin{cases}
        \gamma^\star_i(\hat{K}_d),&i=1,\dots,p,\\
        \mathbf{0},&i=p+1,
    \end{cases}
\end{equation}
iff $\norm{\zeta_{p+1}(\hat{K}^+_d)^\top R(\hat{K}_d)}_2\leq\lambda$. In other words, adding such new elements does not improve the optimal objective function value, or change the transfer function estimate $\hat{G}(q)$. So the new element only reduces the objective function value when $\norm{\zeta_{k_{p+1}}^\top R(\hat{K}_d)}_2>\lambda$. This also guarantees $k_{p+1}\notin\hat{K}_d$ since $\norm{\zeta_{k_i}^\top R(\hat{K}_d)}_2\leq\lambda$ for $i=1,\dots,p$.

Motivated by the above observation, Algorithm~\ref{al:1} is proposed to solve the infinite-dimensional group lasso problem (\ref{eqn:3}), where a greedy strategy is applied that chooses the new element by maximizing $\norm{\zeta_{k_{p+1}}^\top R(\hat{K}_d)}_2$. Note that $\hat{K}^{l}_d$ denotes the set $\hat{K}_d$ at the $l$-th iteration in Algorithm~\ref{al:1}. The transfer function and the pole location estimates $\hat{G}(q)$ and $\hat{S}$ can be calculated by (\ref{eqn:ghat}) and (\ref{eqn:phat}) respectively with discretized atomic set $\hat{K}^l_d$, which is the output of Algorithm~\ref{al:1}.

\begin{algorithm}[tb]
	\caption{A greedy algorithm for the infinite-dimensional group lasso problem (\ref{eqn:3})}
	\begin{algorithmic}[1]
	\State \textbf{Input:} identification data $(\mathbf{u},\mathbf{y})$, $\epsilon>0$, $l_\text{max}$
	\State Initialize $\hat{K}^0_d=\left\{k_1,k_2,\dots,k_{p_0}\right\}$.
	\State Calculate $\Gamma^\star(\hat{K}^0_d)$.
    \State $l\leftarrow 0$
	\Repeat
	\State Construct a candidate new atom
	\begin{equation}
	    k^+\leftarrow\underset{k\in\hat{K}}{\text{argmax}}\ \norm{\zeta_k^\top R(\hat{K}^l_d)}_2.
	    \label{eqn:4}
	\end{equation}
	\If{\norm{\zeta_{k^+}^\top R(\hat{K}^l_d)}_2\geq\lambda+\epsilon}
	\Begin
	\State $k_{p_0+l+1}\leftarrow k^+$, $\hat{K}^{l+1}_d\leftarrow\hat{K}^l_d\cup \{k_{p_0+l+1}\}$
	\State Calculate $\Gamma^\star(\hat{K}^{l+1}_d)$ via program (\ref{eqn:3}).
	\End
	\Else
	\State Break
	\State $l\leftarrow l+1$
	\Until{l\geq l_\text{max}}
	\State \textbf{Output:} $\hat{K}^{l}_d$, $\Gamma^\star(\hat{K}^{l}_d)$
	\end{algorithmic}
	\label{al:1}
\end{algorithm}

Let
\begin{equation}
    \hat{\Gamma}^\star=\left\{\gamma^\star_k\,\left|\,
        \gamma^\star_k=
    \begin{cases}
        \gamma^\star_i(\hat{K}^l_d), &k=k_i\in\hat{K}^l_d\\
        \mathbf{0},&k\in\hat{K}\setminus\hat{K}^l_d
    \end{cases}
    \right.\right\}.
\end{equation}
Algorithm~\ref{al:1} guarantees the following property.
\begin{prop}
If Algorithm~\ref{al:1} terminates without reaching the maximum number of iterations ($l<l_\text{max}$), $\hat{\Gamma}^\star$ satisfies the approximate optimality conditions
\begin{equation}
\begin{cases}
    \norm{\zeta_k^\top R}_2<\lambda+\epsilon,&\text{if }\norm{\gamma^\star_k}_2=0,\\
    \zeta_k^\top R+\lambda\gamma^\star_k/\norm{\gamma^\star_k}_2=0,&\text{if }\norm{\gamma^\star_k}_2> 0,
\end{cases}
\label{eqn:aopt}
\end{equation}
for all $k\in\hat{K}$.
\label{prop:1}
\end{prop}
\begin{proof}
    Since $\gamma^\star_k=0$ for $k\notin\hat{K}^l_d$ in $\hat{\Gamma}^\star$, we have $R=R(\hat{K}^l_d)$. For $k\in\hat{K}^l_d$, the discretized optimality conditions (\ref{eqn:dopt}) guarantee the satisfaction of (\ref{eqn:aopt}). According to Algorithm~\ref{al:1}, $\norm{\zeta_k^\top R(\hat{K}^l_d)}_2=\norm{\zeta_k^\top R}_2<\lambda+\epsilon$. So for $k\notin\hat{K}^l_d$, (\ref{eqn:aopt}) is satisfied since $\norm{\gamma^\star_k}_2=0$.
\end{proof}

Proposition~\ref{prop:1} shows that the infinite-dimensional problem (\ref{eqn:3}) is approximately equivalent to the finite-dimensional problem with $(p_0+l)$ atoms
\begin{equation}
\underset{\left\{\gamma_i\right\}_{i=1}^{p_0+l}}{\text{argmin}}\ \norm{\,\mathbf{y}-\sum_{i=1}^{p_0+l}\zeta_{k_i} \gamma_i}_2^2 + 2\lambda \sum_{i=1}^{p_0+l}\norm{\gamma_i}_2.
\label{eqn:5}
\end{equation}
For the rest of the paper, define $p=p_0+l$.

The main difficulty in Algorithm~\ref{al:1} is solving the non-convex problem (\ref{eqn:4}). However, even if (\ref{eqn:4}) is not solved exactly, Algorithm~\ref{al:1} still guarantees a decrease in the objective function value at each iteration as long as $\norm{\zeta_{k^+}^\top R(\hat{K}^l_d)}_2\geq\lambda+\epsilon$ is satisfied for the candidate atom $k^+$.

\section{Debiasing and Stability Selection}
Algorithm~\ref{al:1} provides a method to solve the group lasso problem (\ref{eqn:3}). However, solutions to lasso-type regularized problems are known to have a large bias and a large number of false positives in feature selection \cite{HighDimensional}. To mitigate these problems, the following tools in high-dimensional statistics are applied to debias the estimate and reject false positives in pole location estimation from Algorithm~\ref{al:1}.

\subsection{Iteratively Reweighted Adaptive Group Lasso}

The $l_1$-norm regularizer (\ref{eqn:lasso}) is a convex relaxation of the ideal sparsity promoting function $J^*(C)=n(S)$, where $n(\cdot)$ denotes the cardinality of the set, which counts the number of poles in the model. Compared to the ideal regularizer which penalizes all the active atoms with a fixed value of 1, the $l_1$-norm regularizer penalizes them with the magnitude of the corresponding coefficients. This induces a negative bias, especially for the atoms with larger coefficients, i.e., the dominant modes. This bias is a large source of error in atomic norm regularization \cite{Pillonetto_2016}.

To reduce such bias, adaptive lasso \cite{doi:10.1198/016214506000000735} has been proposed which adds a second step that applies a reweighted version of the $l_1$-norm regularizer
\begin{equation}
    J_\text{a}(C)=\sum_{k\in K}\frac{\left|c_k\right|}{\left|c^{\star,0}_k\right|+\epsilon'},
\end{equation}
where $c^{\star,0}_k$ is the solution to the original problem, and $\epsilon'>0$ is a small constant to avoid singularity. This regularizer reduces the amount of regularization for atoms estimated with large coefficients in the original problem, and is close to $J^*(C)$ when $c_k\approx c^{\star,0}_k$. This approach is extended to apply this reweighting iteratively (Section 2.8.5 in \cite{HighDimensional}), which is sometimes known as iteratively reweighted lasso. It is pointed out in \cite{5109694} that the iteratively reweighted lasso can be interpreted as a difference of convex programming algorithm to solve the regularized problem with a non-convex log regularizer
\begin{equation}
J_\text{log}(C)=\sum_{k\in K}\frac{\log\left(\left|c_k\right|+\epsilon'\right)}{\log\epsilon'}.
\end{equation}

This iteratively reweighted adaptive approach is applied to the group lasso problem (\ref{eqn:3}) in Algorithm~\ref{al:2}. It is easy to see that the cardinality of the active atomic set $J^*(C)$ is non-increasing at each iteration.

\begin{algorithm}[tb]
	\caption{Iteratively reweighted adaptive group lasso}
	\begin{algorithmic}[1]
	\State \textbf{Input:} identification data $(\mathbf{u},\mathbf{y})$, $\epsilon'>0$, $m_s$
	\State Find
	$
	    \hat{K}^{l}_d=\left\{k_1,\dots,k_{p}\right\},\ \Gamma^\star(\hat{K}^{l}_d):=\left\{\gamma^{\star,0}_1,\dots,\gamma^{\star,0}_{p}\right\}
	$
	from Algorithm~\ref{al:1}.
	\For{m=1}{m_s}
	\Begin
	\State Find $\left\{\gamma_i^{\star,m}\right\}_{i=1}^{p}$ by solving
	\begin{equation}
	\underset{\left\{\gamma_i\right\}_{i=1}^{p}}{\text{argmin}}\ \norm{\,\mathbf{y}-\sum_{i=1}^{p}\zeta_{k_i} \gamma_i}_2^2 + 2\lambda \sum_{i=1}^{p}\frac{\norm{\gamma_i}_2}{\norm{\gamma_i^{\star,m-1}}_2+\epsilon'}.
	\end{equation}
	\End
	\State Calculate $\hat{G}(q)$ by (\ref{eqn:ghat}) with discretized atomic set $\hat{K}^l_d$ and coefficients $\left\{\gamma_i^{\star,m_s}\right\}_{i=1}^{p}$.
	\State \textbf{Output:} $\hat{G}(q)$
	\end{algorithmic}
	\label{al:2}
\end{algorithm}

\subsection{Complementary Pairs Stability Selection}

Lasso-type regularized problems are known to have favorable consistency properties in terms of prediction under mild conditions. However, in terms of estimating the active atomic set $S$, they can only guarantee that the non-active atoms are not in the true model with high probability under practical assumptions (Chapter 2 in \cite{HighDimensional}). This means that the number of false positives in the estimated pole locations is not controlled. In fact, there are usually many more estimated poles than the true poles, with many occurring at ``random'' locations depending on the noise realization. This will be shown in Section~\ref{sec:num}. This phenomenon is known as ``p-value lottery" \cite{doi:10.1198/jasa.2009.tm08647}.

Subsampling techniques have been used to increase the stability of the active atomic set estimation. In particular, the complementary pairs stability selection (CPSS) is applied in this work \cite{https://doi.org/10.1111/j.1467-9868.2011.01034.x}. This method generates complementary pairs of subsamples from the identification data, and repeats the baseline variable selection procedure (group lasso problem (\ref{eqn:3}) here) on each subsample. For our problem, this corresponds to replacing the loss function
with
\begin{equation}
    V_B(\cdot)=\norm{\,\mathbf{y}(B)-\sum_{i=1}^{p}\zeta_{k_i}(B,:) \gamma_i}_2^2,
\end{equation}
where $B\subset\left\{1,2,\dots,N\right\}$ defines a random subsample of data. Define the estimated pole locations on the subsample as $\hat{S}_B$. Then the so-called stable solution of the problem is defined as the atoms that have higher empirical probabilities of being included in $\hat{S}_B$ than a predefined threshold $\tau$. The algorithm has favorable false-positive rejection properties when $\tau>0.5$ \cite{https://doi.org/10.1111/j.1467-9868.2011.01034.x}. The method is summarized in Algorithm~\ref{al:3}. The transfer function can also be estimated by least squares on the stable solution of the atomic set.

\begin{algorithm}[tb]
	\caption{Complementary pairs stability selection}
	\begin{algorithmic}[1]
	\State \textbf{Input:} identification data $(\mathbf{u},\mathbf{y})$, $\tau\in(0.5,1]$, $n_s$
	\State Find $\hat{K}^{l}_d$ from Algorithm~\ref{al:1}.
	\For{i=1}{n_s}
    \Begin
    \State Generate a random subsample $B_i\subset\left\{1,2,\dots,N\right\}$ with $\lfloor N/2\rfloor$ elements.
    \State $\bar{B}_i \leftarrow \left\{1,2,\dots,N\right\}\setminus B_i$
    \State Calculate $\hat{S}_{B_i}$, $\hat{S}_{\bar{B}_i}$ by solving (\ref{eqn:5}) with the loss function $V(\cdot)$ replaced by $V_{B_i}(\cdot)$, $V_{\bar{B}_i}(\cdot)$ respectively.
    \End
    \State $
    \hat{S}\leftarrow\left\{k\,\left|\frac{1}{2n_s}\sum_{i=1}^{n_s}\left(\mathds{1}_{\hat{S}_{B_i}}\!(k)+\mathds{1}_{\hat{S}_{\bar{B}_i}}\!(k)\right)\geq\tau\right.\right\}
    $, where $\mathds{1}$ denotes the indicator function.
	\State \textbf{Output:} $\hat{S}$
	\end{algorithmic}
	\label{al:3}
\end{algorithm}

\section{Numerical Results}
\label{sec:num}
The performances of the proposed algorithms are assessed by numerical simulation on a benchmark fourth-order system previously analyzed in \cite{LANDAU199577}:
\begin{equation}
    G(q)=\dfrac{0.10884q+0.19513}{q^4-1.41833q^3+1.58939q^2-1.31608q+0.88642}.
\end{equation}
The system has been normalized to have an $\mathcal{H}_2$-norm of 1. In what follows, results obtained with Algorithms~\ref{al:1}, \ref{al:2}, and \ref{al:3} are labelled by \textit{InfA}, \textit{AdpInfA}, and \textit{SS} respectively.

Identification data of length $N=100$ are generated with zero-mean i.i.d. unit Gaussian inputs from a zero initial condition. Two noise levels $\sigma^2=0.1$ and 0.01 are considered. The atomic responses $\phi_k$ are also generated from a zero initial condition. 100 Monte Carlo simulations are conducted for each noise level. The initial discretized atomic set $\hat{K}_d^0$ contains $p_0=50$ randomly generated atoms with $k_i=\alpha_i\exp{\left(j\beta_i\right)}$, where $\alpha_i$ and $\beta_i$ are subject to uniform distributions in $[0,1)$ and $[0,\pi]$ respectively. Finite-dimensional group lasso problems are solved by MOSEK. The candidate atom generation problem (\ref{eqn:4}) is solved by the particle swarm solver in \textsc{Matlab}. The hyperparameter $\lambda$ is selected by cross-validation from a 15-point log-space grid between 0.05 and 5 for $\sigma^2=0.1$, and between 0.005 and 0.5 for $\sigma^2=0.01$, except for \textit{SS} where $\lambda$ is fixed to 0.5 for $\sigma^2=0.1$ and 0.05 for $\sigma^2=0.01$. The following parameters are used in simulation: $\epsilon=\epsilon'=10^{-5}$, $\tau=0.9$, $n_s=50$, $m_s=2$.

First, the number of additional atoms $l$ required in Algorithm~\ref{al:1} is plotted against the $\lambda$ values in Figure~\ref{fig:1}. The maximum $l$ in all Monte Carlo simulations is 118, which is below the $l_\text{max}$ setting. Results show that the proposed greedy atom generation approach is able to converge within a reasonable number of iterations, and the required number of additional atoms decreases with $\lambda$.

\begin{figure}[tb]
\centerline{\includegraphics[width=\columnwidth]{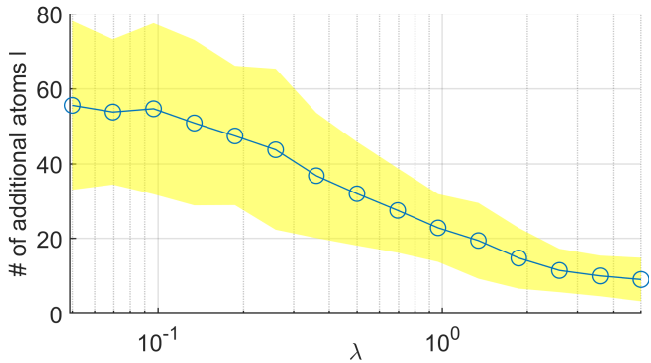}}
\caption{The number of additional atoms $l$ in Algorithm~\ref{al:1} for $\sigma^2=0.1$. Blue: mean values, yellow: ranges within one standard deviation.}
\label{fig:1}
\end{figure}

To demonstrate the performance of the proposed algorithms, they are compared to three benchmark algorithms: 1) least-squares estimation with an ARX model and a known model order (\textit{ARX}); 2) kernel-based identification with a TC kernel design (\textit{TCK}) \cite{Chen_2012}
; 3) discretized atomic norm regularization in \cite{6426006} with 50 (\textit{Atom}) and 500 (\textit{Atom2}) random atoms. Note that \textit{Atom2} uses a significantly larger atomic set compared to Algorithm~\ref{al:1}, as shown in Figure~\ref{fig:1}.

Figure~\ref{fig:2} compares the identification accuracy of all algorithms in terms of the impulse response fitting, defined as
\begin{equation}
    W=100\cdot \left(1-\left[\frac{\sum_{i=1}^{N-1}(g_i-\hat{g}_i)^2}{\sum_{i=1}^{N-1}(g_i-\bar{g})^2}\right]^{1/2}\right),
    \label{eq:W}
\end{equation}
where $g_i$ are the true impulse responses, $\hat{g}_i$ are the estimated impulse responses, and $\bar{g}$ is the mean of $g_i$. It can be seen that the three proposed algorithms all perform better than the benchmark algorithms at both noise levels. In particular, \textit{InfA} obtains better fitting compared to \textit{Atom2} which uses a much larger atomic set. This demonstrates the effectiveness of the proposed atom generation approach. \textit{AdpInfA} further improves on the identification accuracy of \textit{InfA} with iterative reweighting.

\begin{figure}[tb]
\centerline{\includegraphics[width=\columnwidth]{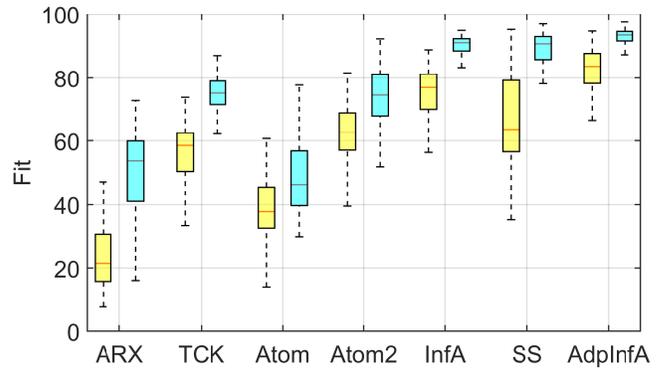}}
\caption{Boxplot of impulse response fitting. Yellow: $\sigma^2=0.1$, cyan: $\sigma^2=0.01$.}
\label{fig:2}
\end{figure}

To further investigate the sources of the estimation errors, Table~\ref{tbl:1} shows the bias-variance analysis of impulse response estimation. As an algorithm proposed to debias the estimate, \textit{AdpInfA} indeed produces a much smaller bias compared to all other algorithms. This is also the main contributor to the reduction of MSE compared to the baseline \textit{InfA} algorithm, on which \textit{AdpInfA} is based.

\begin{table}[tb]
	\renewcommand{\arraystretch}{1.2}
	\caption{Bias-variance analysis of impulse response estimation}
	\label{tbl:1}
	\resizebox{\linewidth}{!}{
	\centerline{
		\begin{tabular}{ccccccc}
			\hline\hline
			& \textbf{TCK} & \textbf{Atom} & \textbf{Atom2} & \textbf{InfA} & \textbf{SS} & \textbf{AdpInfA} \\ \hline 
			\textbf{$\sigma^2=0.1$}&&&&&& \\
			\textbf{$\text{Bias}^2\ [\times 10^{-2}]$} &  6.76 & 23.42 &  6.34 & 2.63 &  8.28 & 0.91\\
			\textbf{$\text{Var}\ [\times 10^{-2}]$}    & 13.04 & 13.59 &  8.52 & 3.80 & 15.68 & 2.70\\ 
			\textbf{$\text{MSE}\ [\times 10^{-2}]$}    & 19.80 & 37.01 & 14.86 & 6.44 & 23.96 & 3.60\\ \hline
			\textbf{$\sigma^2=0.01$}&&&&&& \\
			\textbf{$\text{Bias}^2\ [\times 10^{-2}]$} &  1.78 & 15.92 & 2.22 & 0.43 & 0.47 & 0.07\\
			\textbf{$\text{Var}\ [\times 10^{-2}]$}    &  5.45 & 11.68 & 5.26 & 0.76 & 3.12 & 0.52\\
			\textbf{$\text{MSE}\ [\times 10^{-2}]$}    &  7.23 & 27.60 & 7.48 & 1.18 & 3.59 & 0.59\\ \hline\hline
		\end{tabular}
	}}
\end{table}

Finally, the capability of estimating the poles of the system is demonstrated in Figures~\ref{fig:3} and \ref{fig:4}. It is illustrated in Figure~\ref{fig:3} that all the algorithms that directly solve group lasso problems estimate a much larger number of poles compared to the true one. \textit{AdpInfA} mitigates the over-estimation since the active atomic set shrinks at each iteration, whereas $\text{SS}$ obtains a very accurate estimation of the model order.

\begin{figure}[tb]
\centerline{\includegraphics[width=\columnwidth]{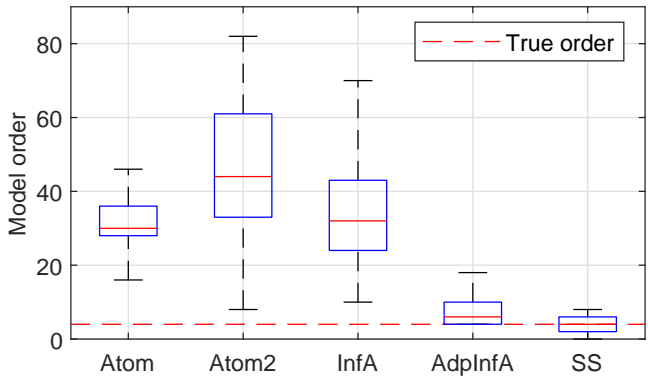}}
\caption{Comparison of estimated model orders for $\sigma^2=0.1$.}
\label{fig:3}
\end{figure}

To assess the accuracy of pole location estimation, Figure~\ref{fig:4} further compares the distributions of estimated pole locations in all 100 Monte Carlo simulations. Despite knowing the true model order, \textit{ARX} fails to give accurate estimations of the pole locations. Although the estimated model order is close to the true one, \textit{AdpInfA} estimates a significant number of false positives in terms of the actual pole locations. Among all the algorithms, only \textit{SS} is able to obtain accurate pole location estimations with few false positives, which proves the effectiveness of the CPSS method.

\begin{figure}[tb]
\centerline{\includegraphics[width=\columnwidth]{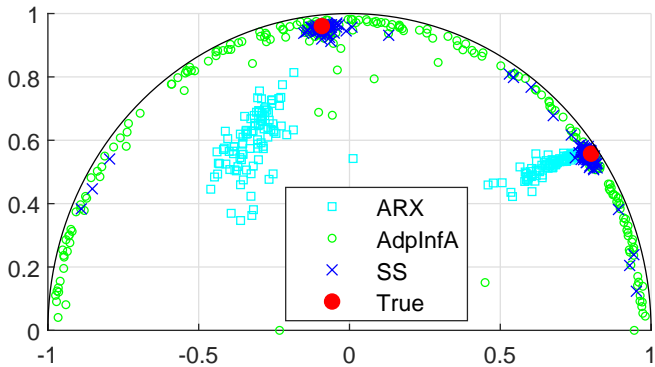}}
\caption{Comparison of pole location estimation distributions in all 100 Monte Carlo simulations for $\sigma^2=0.1$.}
\label{fig:4}
\end{figure}

\section{Conclusions}
This work applies advanced techniques studied in high-dimensional statistics to the atomic norm regularization problem in linear system identification. A greedy algorithm is presented to generate new candidate atomic models from infinitely many possible pole locations. Common drawbacks of lasso-type regularization are mitigated by adaptively adjusting the regularization weights for each atom and selecting only repeatedly occurring pole locations from subsamples of data. Results in this paper suggest that sparse learning algorithms are a promising alternative to kernel-based methods with fewer design requirements and direct pole location estimation. Further research directions include improvements in computational efficiency, comparison with model order reduction methods, and extensions to MIMO systems and frequency-domain data.


\bibliographystyle{IEEEtran}
\bibliography{refs}

\begin{thebibliography}{10}
\providecommand{\url}[1]{#1}
\csname url@samestyle\endcsname
\providecommand{\newblock}{\relax}
\providecommand{\bibinfo}[2]{#2}
\providecommand{\BIBentrySTDinterwordspacing}{\spaceskip=0pt\relax}
\providecommand{\BIBentryALTinterwordstretchfactor}{4}
\providecommand{\BIBentryALTinterwordspacing}{\spaceskip=\fontdimen2\font plus
\BIBentryALTinterwordstretchfactor\fontdimen3\font minus
  \fontdimen4\font\relax}
\providecommand{\BIBforeignlanguage}[2]{{%
\expandafter\ifx\csname l@#1\endcsname\relax
\typeout{** WARNING: IEEEtran.bst: No hyphenation pattern has been}%
\typeout{** loaded for the language `#1'. Using the pattern for}%
\typeout{** the default language instead.}%
\else
\language=\csname l@#1\endcsname
\fi
#2}}
\providecommand{\BIBdecl}{\relax}
\BIBdecl

\bibitem{LjungBook2}
L.~Ljung, \emph{System Identification: Theory for the User}.\hskip 1em plus
  0.5em minus 0.4em\relax Upper Saddle River, NJ, USA: Prentice-Hall, 1999.

\bibitem{ASTROM1980551}
K.~Åström, ``Maximum likelihood and prediction error methods,''
  \emph{Automatica}, vol.~16, no.~5, pp. 551--574, 1980.

\bibitem{Ljung_2019}
L.~Ljung, ``On convexification of system identification criteria,''
  \emph{Automation and Remote Control}, vol.~80, no.~9, pp. 1591--1606, 2019.

\bibitem{Ljung2_2019}
L.~Ljung, T.~Chen, and B.~Mu, ``A shift in paradigm for system
  identification,'' \emph{International Journal of Control}, vol.~93, no.~2,
  pp. 173--180, 2019.

\bibitem{Pillonetto_2016}
G.~Pillonetto, T.~Chen, A.~Chiuso, G.~D. Nicolao, and L.~Ljung, ``Regularized
  linear system identification using atomic, nuclear and kernel-based norms:
  the role of the stability constraint,'' \emph{Automatica}, vol.~69, pp.
  137--149, 2016.

\bibitem{PILLONETTO201081}
G.~Pillonetto and G.~{De Nicolao}, ``A new kernel-based approach for linear
  system identification,'' \emph{Automatica}, vol.~46, no.~1, pp. 81--93, 2010.

\bibitem{PILLONETTO2014657}
G.~Pillonetto, F.~Dinuzzo, T.~Chen, G.~{De Nicolao}, and L.~Ljung, ``Kernel
  methods in system identification, machine learning and function estimation: a
  survey,'' \emph{Automatica}, vol.~50, no.~3, pp. 657--682, 2014.

\bibitem{CHEN2018109}
T.~Chen, ``On kernel design for regularized {LTI} system identification,''
  \emph{Automatica}, vol.~90, pp. 109--122, 2018.

\bibitem{Chen_2012}
T.~Chen, H.~Ohlsson, and L.~Ljung, ``On the estimation of transfer functions,
  regularizations and gaussian processes{\textemdash}revisited,''
  \emph{Automatica}, vol.~48, no.~8, pp. 1525--1535, 2012.

\bibitem{Smith_2014}
R.~S. Smith, ``Frequency domain subspace identification using nuclear norm
  minimization and {Hankel} matrix realizations,'' \emph{{IEEE} Transactions on
  Automatic Control}, vol.~59, no.~11, pp. 2886--2896, 2014.

\bibitem{945730}
M.~Fazel, H.~Hindi, and S.~Boyd, ``A rank minimization heuristic with
  application to minimum order system approximation,'' in \emph{Proceedings of
  the 2001 American Control Conference. (Cat. No.01CH37148)}, vol.~6, 2001, pp.
  4734--4739.

\bibitem{6426006}
P.~Shah, B.~N. Bhaskar, G.~Tang, and B.~Recht, ``Linear system identification
  via atomic norm regularization,'' in \emph{51st IEEE Conference on Decision
  and Control (CDC)}, 2012, pp. 6265--6270.

\bibitem{Yuan_2006}
M.~Yuan and Y.~Lin, ``Model selection and estimation in regression with grouped
  variables,'' \emph{Journal of the Royal Statistical Society: Series B
  (Statistical Methodology)}, vol.~68, no.~1, pp. 49--67, 2006.

\bibitem{YIN20201237}
M.~Yin, A.~Iannelli, M.~Khosravi, A.~Parsi, and R.~S. Smith, ``Linear
  time-periodic system identification with grouped atomic norm
  regularization,'' \emph{IFAC-PapersOnLine}, vol.~53, no.~2, pp. 1237--1242,
  2020, 21st IFAC World Congress.

\bibitem{KHOSRAVI2020412}
M.~Khosravi, M.~Yin, A.~Iannelli, A.~Parsi, and R.~S. Smith, ``Low-complexity
  identification by sparse hyperparameter estimation,''
  \emph{IFAC-PapersOnLine}, vol.~53, no.~2, pp. 412--417, 2020, 21st IFAC World
  Congress.

\bibitem{doi:10.1198/jasa.2009.tm08647}
N.~Meinshausen, L.~Meier, and P.~Bühlmann, ``p-values for high-dimensional
  regression,'' \emph{Journal of the American Statistical Association}, vol.
  104, no. 488, pp. 1671--1681, 2009.

\bibitem{Rakotomamonjy_2012}
A.~Rakotomamonjy, R.~Flamary, and F.~Yger, ``Learning with infinitely many
  features,'' \emph{Machine Learning}, vol.~91, no.~1, pp. 43--66, 2012.

\bibitem{10.1007/978-3-540-72927-3_39}
S.~Rosset, G.~Swirszcz, N.~Srebro, and J.~Zhu, ``$l_1$ regularization in
  infinite dimensional feature spaces,'' in \emph{Learning Theory}.\hskip 1em
  plus 0.5em minus 0.4em\relax Berlin, Heidelberg: Springer, 2007, pp.
  544--558.

\bibitem{NIPS2014_459a4ddc}
I.~E.~H. Yen, T.~W. Lin, S.~D. Lin, P.~K. Ravikumar, and I.~S. Dhillon,
  ``Sparse random feature algorithm as coordinate descent in {Hilbert} space,''
  in \emph{Advances in Neural Information Processing Systems}, vol.~27, 2014.

\bibitem{WANG20085277}
H.~Wang and C.~Leng, ``A note on adaptive group lasso,'' \emph{Computational
  Statistics \& Data Analysis}, vol.~52, no.~12, pp. 5277--5286, 2008.

\bibitem{5109694}
G.~Gasso, A.~Rakotomamonjy, and S.~Canu, ``Recovering sparse signals with a
  certain family of nonconvex penalties and {DC} programming,'' \emph{IEEE
  Transactions on Signal Processing}, vol.~57, no.~12, pp. 4686--4698, 2009.

\bibitem{HighDimensional}
P.~Bühlmann and S.~van~de Geer, \emph{Statistics for high-dimensional data:
  methods, theory and applications}.\hskip 1em plus 0.5em minus 0.4em\relax
  Berlin, Heidelberg: Springer, 2011.

\bibitem{https://doi.org/10.1111/j.1467-9868.2011.01034.x}
R.~D. Shah and R.~J. Samworth, ``Variable selection with error control: another
  look at stability selection,'' \emph{Journal of the Royal Statistical
  Society: Series B (Statistical Methodology)}, vol.~75, no.~1, pp. 55--80,
  2013.

\bibitem{10.2307/2346178}
R.~Tibshirani, ``Regression shrinkage and selection via the lasso,''
  \emph{Journal of the Royal Statistical Society. Series B (Methodological)},
  vol.~58, no.~1, pp. 267--288, 1996.

\bibitem{doi:10.1198/016214506000000735}
H.~Zou, ``The adaptive lasso and its oracle properties,'' \emph{Journal of the
  American Statistical Association}, vol. 101, no. 476, pp. 1418--1429, 2006.

\bibitem{LANDAU199577}
I.~Landau, D.~Rey, A.~Karimi, A.~Voda, and A.~Franco, ``A flexible transmission
  system as a benchmark for robust digital control,'' \emph{European Journal of
  Control}, vol.~1, no.~2, pp. 77--96, 1995.

\end{thebibliography}

\end{document}